%% file: main.tex
\documentclass[journal, onecolumn]{IEEETran}
\usepackage{color,graphicx}
\usepackage{amsmath,amssymb,amsfonts,amsthm}
\usepackage{mathtools}
\usepackage{microtype}
\usepackage{booktabs}
\usepackage[inline]{enumitem}
\usepackage{svg}
\usepackage{subcaption}
\usepackage{etoolbox}
\usepackage{gnuplottex}
\usepackage{xcolor}
\usepackage{tikz}
\usetikzlibrary{decorations.pathreplacing,calc}
\usepackage{hyperref}
\usepackage[linesnumbered,ruled,vlined]{algorithm2e}

\newcommand{\wmin}{w_{\min}}
\newcommand{\wlo}{w^{-}}
\newcommand{\whi}{w^{+}}

\newcommand{\wrep}{\widehat{w}}
\newcommand{\vrep}{\widehat{v}}
\newcommand{\rhorep}{\widehat{\rho}}
\newcommand{\Gc}{\mathcal{G}}
\newcommand{\Zc}{\mathcal{Z}}
\newcommand{\Gm}{\Gamma}
\newcommand{\wf}{\mathrm{wf}}
\newcommand{\gp}{\mathrm{gp}}

\theoremstyle{plain}
\newtheorem{theorem}{Theorem}
\newtheorem{lemma}[theorem]{Lemma}
\newtheorem{proposition}[theorem]{Proposition}

\theoremstyle{definition}
\newtheorem{definition}[theorem]{Definition}
\newtheorem{assumption}[theorem]{Assumption}

\theoremstyle{remark}

\providecommand{\IEEEPARstart}[2]{#1#2}
\providecommand{\IEEEtriggeratref}[1]{}

\allowdisplaybreaks
\setlist{topsep=2pt,itemsep=1pt,parsep=0pt}

\begin{document}

% \begin{frontmatter}

\title{A Group-Based Resource Allocation Model for the Fractional Knapsack Problem}

% \author[1]{Abhinaba Chakraborty\corref{cor1}}
% % \ead{chakraborty@ieee.org}
% % \cortext[cor1]{Corresponding author}
% \affiliation[1]{
%     organization={ID Lab, University of Ghent - imec}, 
%     % addressline={}, 
%     city={Ghent}, 
%     postcode={9052}, 
%     country={Belgium}
% }
\author{
\IEEEauthorblockN{Abhinaba Chakraborty\IEEEauthorrefmark{1}, 
                  % Wouter Tavernier\IEEEauthorrefmark{1},
                  % Mario Pickavet\IEEEauthorrefmark{1}, 
                  % Didier Colle\IEEEauthorrefmark{1}
                  }\\
\IEEEauthorrefmark{1}ID Lab, University of Ghent-imec, Ghent, Belgium}
\maketitle
\begin{abstract}
\input{utils/abstract}
\end{abstract}

% \begin{keyword}
% Fractional knapsack problem \sep Continuous knapsack \sep Water-filling \sep Axiomatic fair division 
% \end{keyword}

% \end{frontmatter}

% \input{utils/problem}

\input{utils/motivation}

\input{utils/algorithm}

\input{utils/guarantees}

\input{utils/related_works}

\input{utils/Evaluation}

\bibliographystyle{elsarticle-num}
\bibliography{refs}

\end{document}

%% file: utils/abstract.tex
\noindent
To solve the fractional knapsack problem, Dantzig's greedy rule orders items according to their value-to-cost ratio. This ordering introduces priority issues. An arbitrarily small perturbation to the input can change the allocation if the budget is exhausted between two items with very similar ratios. To mitigate that problem, we introduce a two-stage rule. We group items sharing attributes within a radius $\delta$. We then evaluate these groups in descending order of ratio and divide their group's budget share without further ranking. Consider a group featuring an aggregate capacity $U_G$, unit costs contained in $[w^-,w^+]$, and a representative value $\widehat{v}$. The group's loss relative to the exact optimum is bounded by $\widehat{v}\, U_G\frac{w^+-w^-}{w^++w^-}+\varepsilon_v U_G$, where $\varepsilon_v$ limits the group's internal value variation. Moreover, this harmonic factor remains tight for any group size. The overall loss becomes restricted to the single budget-binding group whenever the grouping remains order-compatible; thus, groups containing at most $K$ items suffer a per-item loss of $\mathcal{O}(\frac{K}{n})$. Should group ratio intervals exhibit an overlap of at most $\omega$, an additive term $\omega C$ degrades this bound. Within the separation margin between adjacent groups, the grouped allocation remains Lipschitz continuous with respect to cost data, exhibiting a modulus of $\frac{K}{w_{\min}}$. Computing this allocation takes $\mathcal{O}(n+m\log m+|\Gamma|\log|\Gamma|)$ time given $m$ groups and a boundary group $\Gamma$. Alternatively, the time complexity drops to $\mathcal{O}(n+m\log m)$ if a linear-time selection method identifies the boundary group's allocation.

%% file: utils/motivation.tex
\section{Motivation: Discontinuity and Spurious Ranking}\label{sec:motivation}
A set of $n$ items defines an instance $I=(v,w,u,C)$, where items are indexed by $[n]\coloneqq\{0,\dots,n-1\}$. Every item has a value $v_i>0$, a cost $w_i>0$, and a capacity limit $u_i\in[0,1]$. A total budget $C\ge0$ is also provided. Given a subset $S\subseteq[n]$, we denote its total capacity as $U_S\coloneqq\sum_{i\in S}u_i$, whereas $F^w_S\coloneqq\sum_{i\in S}w_iu_i$. We refer to $\rho_i\coloneqq \frac{v_i}{w_i}$ as the efficiency ratio for item $i$. Thus, the bounded fractional knapsack problem can be stated as follows:
\begin{gather}
    V^\star(I)\coloneqq\max_{z\in\Zc(I)}\ \sum_{i\in[n]}v_iz_i,\nonumber \\
    \Zc(I)\coloneqq\Bigl\{z\in\prod_{i\in[n]}[0,u_i]:\ \sum_{i\in[n]}w_iz_i\le C\Bigr\}
    \label{eq:P}
\end{gather}
Throughout this article, we assume $C<F^w_{[n]}$; hence, the budget cannot saturate all items simultaneously. We will use the notation $V^\star(C)$ in contexts where $(v,w,u)$ remain fixed while the budget varies.

\begin{lemma}[Optimality conditions]\label{lem:price}
A given feasible allocation $z^\star\in\Zc(I)$ becomes an optimal solution to~\eqref{eq:P} if and only if $\exists\,\nu^\star\ge0$ satisfying $\nu^\star\bigl(C-\sum_iw_iz^\star_i\bigr)=0$ such that the following condition holds:
\begin{equation}
    z^\star_i=
    \begin{cases}
        u_i,\ \rho_i>\nu^\star,\\
        0,\ \rho_i<\nu^\star,\\
        \in[0,u_i],\ \rho_i=\nu^\star
    \end{cases}
    \label{eq:bang-bang}
\end{equation}
Furthermore, on the interval $[0, F^w_{[n]}]$, $V^\star(C)$ behaves as a strictly increasing, piecewise linear, and concave function, satisfying $V^\star(C')\le V^\star(C)+\nu^\star(C'-C)$ for any $C'\in[0, F^w_{[n]}]$; thus, $V^\star(C)\ge\nu^\star C$.
\end{lemma}

\begin{proof}
Let us fix some $\nu\ge0$ along with a $z^\star\in\Zc(I)$ that satisfies both~\eqref{eq:bang-bang} for $\nu^\star=\nu$ and the complementary slackness condition $\nu\bigl(C-\sum_iw_iz^\star_i\bigr)=0$. Because $w_i>0$, the expression $v_i-\nu w_i$ necessarily shares the same sign as $\rho_i-\nu$. Consequently, condition~\eqref{eq:bang-bang} simply dictates that $z^\star_i$ must maximise the product $(v_i-\nu w_i)z_i$ within the interval $[0,u_i]$. Therefore, taking any arbitrary $z$ subject to $0\le z_i\le u_i$ and $\sum_iw_iz_i\le C'$, we obtain
\begin{equation}
    \sum_iv_iz_i\ \le\ \nu C'+\sum_i(v_i-\nu w_i)z_i\ \le\ \nu C'+\sum_i(v_i-\nu w_i)z^\star_i\ =\ \sum_iv_iz^\star_i+\nu(C'-C),
    \label{eq:lagrange-chain}
\end{equation}
Notice that the final equality follows from the assumption that $\nu\bigl(C-\sum_iw_iz^\star_i\bigr)=0$. Setting $C'=C$ immediately demonstrates the optimality of $z^\star$. Following this, the relation $V^\star(C')\le V^\star(C)+\nu(C'-C)$ naturally emerges as the required supergradient inequality.

To prove the converse, we begin by explicitly constructing an allocation adhering to this structure. First, sort and relabel the items ensuring that $\rho_0\ge\dots\ge\rho_{n-1}$. Next, identify $p$ as the smallest index satisfying $\sum_{i\le p}w_iu_i>C$; such an index is guaranteed to exist due to the condition $C<F^w_{[n]}$. We then define $z^g_i=u_i$ for all $i<p$, assign $z^g_p=\bigl(C-\sum_{i<p}w_iu_i\bigr)/w_p\in[0,u_p)$, and enforce $z^g_i=0$ for any $i>p$. This allocation $z^g$ fully consumes the budget $C$ while satisfying~\eqref{eq:bang-bang} by setting $\nu^\star=\rho_p>0$, confirming its optimality. Now, consider an arbitrary optimal solution $z$. We can invoke~\eqref{eq:lagrange-chain} using $z^\star=z^g$, $C'=C$, and $\nu=\rho_p$. Because the expressions at both ends evaluate to $V^\star(C)$, the intermediate inequalities must actually be exact equalities. The first equality implies $\rho_p\bigl(C-\sum_iw_iz_i\bigr)=0$, demonstrating that $z$ must exhaust the entire budget. The second equality applies on a term-by-term basis, meaning $(v_i-\rho_pw_i)z_i=(v_i-\rho_pw_i)z^g_i$. This equation strictly requires $z_i=u_i$ whenever $\rho_i>\rho_p$, and similarly $z_i=0$ whenever $\rho_i<\rho_p$. Consequently, we obtain exactly condition~\eqref{eq:bang-bang} with $\nu^\star=\rho_p$. Lastly, setting $\nu^\star=0$ would be contradictory in~\eqref{eq:bang-bang}. Given that all $\rho_i>0$, this setting would require $z^\star=u$, producing an infeasible allocation.

Our procedure for building $z^g$ likewise characterises the function $V^\star$. Specifically, whenever $C$ falls within the interval $\bigl[\sum_{i<p}w_iu_i,\sum_{i\le p}w_iu_i\bigr]$, the corresponding optimal value evaluates to $\sum_{i<p}v_iu_i+\rho_p\bigl(C-\sum_{i<p}w_iu_i\bigr)$. This expression is affine and has a strictly positive slope of $\rho_p>0$. Furthermore, the sequence of slopes $\rho_0\ge\rho_1\ge\cdots$ is non-increasing. These properties collectively ensure that $V^\star$ remains concave, strictly increasing, and piecewise linear. Applying the supergradient inequality at $C'=0$, combined with the fact that $V^\star(0)=0$, yields the lower bound $V^\star(C)\ge\nu^\star C$.
\end{proof}

\noindent
The specific allocation $z^g$ detailed during the proof corresponds to the well-known Dantzig greedy strategy~\cite{dantzig1957}. This approach involves funding the items in descending order of their efficiency ratios $\rho_i$ until the budget is exhausted. In terms of computational complexity, the procedure requires $\mathcal{O}(n\log n)$ operations if it uses direct sorting. Alternatively, one can achieve $\mathcal{O}(n)$ execution time by relying on weighted median selection techniques~\cite{balas1980,megiddo1983,dyer1984,zemel1984}. Assuming all efficiency ratios are distinct, at most one item will be assigned a fractional quantity, while all remaining coordinates must be set to either $0$ or $u_i$. Consequently, the resulting solution always corresponds to a vertex of the feasible set $\Zc(I)$. This vertex is dictated purely by the relative ordering of the ratios, not their precise numerical values, which ultimately introduces inherent instability.

\begin{proposition}[No Lipschitz modulus]\label{prop:eps}
Assume we are given parameters $\tilde w,\varepsilon>0$, values $v=(1,1)$, capacities $u=(1,1)$, and a budget limit $C=\tilde w$. We analyze two distinct cost vectors: $w=(\tilde w,\tilde w+\varepsilon)$ alongside $w'=(\tilde w+\varepsilon,\tilde w)$. Under these conditions, the uniquely optimal allocation vectors evaluate to $z^\star(w)=(1,0)$ and $z^\star(w')=(0,1)$, which implies
\begin{equation}
    \frac{\|z^\star(w)-z^\star(w')\|_1}{\|w-w'\|_\infty}\ =\ \frac{2}{\varepsilon}\ \xrightarrow[\varepsilon\downarrow0]{}\ \infty .
    \label{eq:lp-quotient-div}
\end{equation}
Because every norm on $\mathbb{R}^2$ is topologically equivalent, the mapping to the optimal solution cannot have a finite Lipschitz constant under any combination of norms.
\end{proposition}

\begin{proof}
Suppose that the pair $(t,s)$ represents a feasible assignment under the cost vector $w$. The underlying budget restriction, given by $\tilde wt+(\tilde w+\varepsilon)s\le\tilde w$, forces the upper bound $s\le\tilde w(1-t)/(\tilde w+\varepsilon)$. This naturally leads to the inequality
\[
    t+s\ \le\ \frac{\tilde w+\varepsilon t}{\tilde w+\varepsilon}\ \le\ 1,
\]
where the sequence of inequalities achieves equality if and only if we have $t=1$ accompanied by $s=0$. Thus, $(1,0)$ is the unique optimal choice for the cost configuration $w$. By applying symmetric reasoning, we conclude that $(0,1)$ must be the strictly unique optimum when confronting $w'$.
\end{proof}

\noindent
This abrupt transition in the solution space is resilient against minor regularisation techniques or arbitrary tie-breaking rules. The optimal configuration immediately shifts between extreme vertices of $\Zc(I)$ the instant the relative ordering between $\rho_0$ and $\rho_1$ reverses, regardless of how minute that reversal might be. To ensure our chosen allocation reflects genuine variation in the dataset rather than measurement noise, our strategy must strictly avoid ranking items whose ratio disparity falls below the underlying data resolution. Consequently, Section~\ref{sec:algorithm} operationalises this concept by bundling items whenever their attribute distance falls within a margin $\delta$, allowing internal group allocations that bypass strict ranking. Subsequently, Section~\ref{sec:guarantees} formally quantifies both the efficiency sacrifice and the corresponding improvements in systemic stability.

%% file: utils/algorithm.tex
\section{The Grouped Allocation Algorithm}\label{sec:algorithm}
\noindent
Initially, we allocate shares to each group by evaluating them in descending order of their "representative" ratio. Then, we divide each allocation among the group's members without further ranking.

\begin{assumption}[Lipschitz data]\label{as:data}
Every item $i\in[n]$ possesses an attribute $a_i$ within a metric space $(A,d)$. Furthermore, functions $v\colon A\to(0,v_{\max}]$ and $w\colon A\to[\wmin,\infty)$ exist, where $\wmin>0$. These functions maintain Lipschitz continuity with constants $L_v$ and $L_w$, satisfying the conditions $v_i=v(a_i)$ and $w_i=w(a_i)$ across all $i$.
\end{assumption}

\begin{definition}[Metric grouping]\label{def:grouping}
We define a \emph{metric grouping} with a tolerance of $\delta>0$ as a partition $\Gc=\{G_0,\dots,G_{m-1}\}$ that divides $[n]$ into non-empty subsets. This partition includes representatives $\widehat a_k\in A$ ensuring that distance, $d(a_i,\widehat a_k)\le\delta$ holds true for every $i\in G_k$. Consequently, the maximum distance between any two items within the same group is bounded by $2\delta$. The terms $\vrep_k\coloneqq v(\widehat a_k)$, $\wrep_k\coloneqq w(\widehat a_k)$, and $\rhorep_k\coloneqq\vrep_k/\wrep_k$ denote the representative value, cost, and ratio, respectively.Define  $\wlo_k\coloneqq\min_{i\in G_k}w_i$ and $\whi_k\coloneqq\max_{i\in G_k}w_i$. For any $i\in G_k$, the following bounds:
\begin{equation}
    \mid v_i-\vrep_k\mid\le L_v\delta\coloneqq\varepsilon_v,\ \whi_k-\wlo_k\le 2L_w\delta\coloneqq\varepsilon_w
    \label{eq:eps-def}
\end{equation}
If the context involves only a single group $G$, we omit the subscript $k$.
\end{definition}

% \subsection{The intra-group rule}\label{sec:axioms}
\noindent
Consider a specific group $G$ alongside an allocation share $c\in[0,F^w_G]$. Since the variation among group members remains below $\delta$, the allocation mechanism dividing $c$ must rely entirely on clearing the assigned share and respecting individual ceilings.

\begin{definition}[Intra-group rules]\label{def:axioms}
An allocation $z\in\mathbb{R}^{G}$ satisfies the intra-group rules for the share $c$ if
\begin{enumerate}[label=\textbf{F\arabic*}]
    \item (Feasibility) $0\le z_i\le u_i$ for all $i\in G$;
    \item (Budget clearance) $\sum_{i\in G}w_iz_i=c$;
    \item (Equal protection) for all $i,j\in G$, if $z_i<z_j$ then $z_i=u_i$.
\end{enumerate}
\end{definition}
\noindent
By Rule \textbf{F3}, an item can be allocated less than a peer only if its personal ceiling restricts further allocation.

\begin{theorem}[Water-filling]\label{th:characterise}
An allocation $z\in\mathbb{R}^G$ satisfies \textbf{F1} and \textbf{F3} if and only if $\exists\,\zeta\ge0$ with
\begin{align}
    z_i=\min\{u_i,\zeta\}\ \forall i\in G .
    \label{eq:wf-rule}
\end{align}
For every $c\in[0,F^w_G]$ exactly one allocation satisfies \textbf{F1}-\textbf{F3}. It is given by~\eqref{eq:wf-rule} with $\zeta$ the unique solution in $[0,\bar u]$, $\bar u\coloneqq\max_{i\in G}u_i$, of
\begin{equation}
    W_G(\zeta)\coloneqq\sum_{i\in G}w_i\min\{u_i,\zeta\}\ =\ c .
    \label{eq:wg-def}
\end{equation}
\end{theorem}

\begin{proof}
Assuming $z$ takes the form given in~\eqref{eq:wf-rule}, condition \textbf{F1} follows trivially. Furthermore, $z_i<z_j\le\zeta\Rightarrow\min\{u_i,\zeta\}<\zeta$, yielding $z_i=u_i$, which satisfies \textbf{F3}. 

\noindent
In the opposite direction, suppose $z$ meets both \textbf{F1} and \textbf{F3}. Setting $\zeta\coloneqq\max_{j\in G}z_j$, if $z_i=\zeta$, condition \textbf{F1} implies $\zeta\le u_i$, thus making $z_i=\min\{u_i,\zeta\}$. Alternatively, if $z_i<\zeta$, we apply \textbf{F3} between $i$ and some index $j$ that reaches the maximum value. This gives $z_i=u_i<\zeta$, meaning $z_i=\min\{u_i,\zeta\}$ holds in this case as well.

Regarding the second assertion, the function $W_G$ remains continuous, starting at $W_G(0)=0$ and reaching $W_G(\bar u)=F^w_G$. It also increases strictly over the interval $[0,\bar u]$. Specifically, taking $0\le\xi<\xi'\le\bar u$, any item $j$ where $u_j=\bar u$ will add a positive amount $w_j(\xi'-\xi)>0$ to the difference $W_G(\xi')-W_G(\xi)$, while no other terms decrease. Therefore, the equation $W_G(\zeta)=c$ possesses a single unique root $\zeta\in[0,\bar u]$. Because any allocation $z$ fulfilling \textbf{F1}--\textbf{F3} must match the form of~\eqref{eq:wf-rule}, substituting $\zeta$ with $\min\{\zeta,\bar u\}$ leaves $z$ unaffected. Thus, we can safely assume $\zeta\in[0,\bar u]$. From there, \textbf{F2} simplifies to $W_G(\zeta)=c$, which directly identifies both $\zeta$ and $z$. Working backwards, solving~\eqref{eq:wg-def} yields an allocation that satisfies all three constraints. This specific allocation is referred to as the \emph{water-filling} allocation and denoted by $z^\wf(c)$.
% It is impossible to weaken Rule \textbf{F3} to merely require equal treatment for items that remain strictly below their ceilings. For example, given parameters $w=(1,1)$, $u=(0.3,1)$, and $c=0.5$, an assignment of $(0.3,0.2)$ fulfills \textbf{F1}, \textbf{F2}, as well as the proposed weaker condition (trivially). However, it completely satisfies the demand of item~$0$ while leaving item~$1$ underfunded, despite both sharing the same cost. In contrast, applying the unmodified Rule \textbf{F3} correctly sets $\zeta=0.25$, yielding the allocation $(0.25,0.25)$.
\end{proof}

\noindent
We now index the groups to guarantee that $\rhorep_0\ge\dots\ge\rhorep_{m-1}$. We define $F^w_{\le k}\coloneqq\sum_{j\le k}F^w_{G_j}$, where $F^w_{\le-1}\coloneqq0$. The corresponding boundary index and leftover allocation are given by
\begin{equation}
    k^\star\coloneqq\min\{k\in[m]:F^w_{\le k}>C\},\ c^\star\coloneqq C-F^w_{\le k^\star-1}\in[0,F^w_{G_{k^\star}})
    \label{eq:kstar}
\end{equation}
The condition $C<F^w_{[n]}$ ensures that $k^\star$ is well-defined. We assign the share $c_j=F^w_{G_j}$ to group $G_j$ whenever $j<k^\star$. If $j=k^\star$, the group receives $c_j=c^\star$, whereas groups with $j>k^\star$ get $c_j=0$. Every allocated share is then distributed. We denote the boundary, fully allocated, and fully unallocated groups as,
\[
    \Gm\coloneqq G_{k^\star},\ \Gm^-\coloneqq\bigcup_{j<k^\star}G_j,\  \Gm^+\coloneqq\bigcup_{j>k^\star}G_j
\]
Thus, we express the grouped allocation as
\begin{equation}
    z^\gp_i=
    \begin{cases}
        u_i,\ i\in\Gm^-,\\
        0,\ i\in\Gm^+,\\
        z^\gp|_{\Gm}=z^\wf(c^\star),\ i\in\Gm
    \end{cases}
    \nonumber
\end{equation}
which yields $V^\gp(C)\coloneqq\sum_iv_iz^\gp_i$. Because $F^w_{\le k^\star-1}+c^\star=C$, the allocation $z^\gp$ proves feasible and exhausts the allocation budget entirely. Within any given group, the final distribution relies on $\zeta$. Consequently, items with identical ceilings are assigned identical shares, regardless of their individual ratios. 
% Furthermore, Theorem~\ref{th:stability} demonstrates that the resulting allocation exhibits Lipschitz continuous behavior when subjected to cost perturbations that do not exceed the gap separating adjacent groups.

\subsection{Algorithm and complexity}\label{sec:complexity}
Drawing on Theorem~\ref{th:characterise}, the function $W_\Gm$ is continuous, strictly monotonic, and piecewise linear over $[0,\bar u]$, with breakpoints precisely at the ceilings in $\Gm$. We order these ceilings as $u_{(0)}\le\dots\le u_{(q-1)}$, setting $q\coloneqq|\Gm|$ and $u_{(-1)}\coloneqq0$. Furthermore, we define $S^w_\ell\coloneqq\sum_{r=\ell}^{q-1}w_{(r)}$. Throughout $[u_{(\ell-1)},u_{(\ell)}]$, items $(r)$ with $r<\ell$ reach their caps, while the remaining items do not. Thus,
\begin{equation}
    W_\Gm(\zeta)\ =\ W_\Gm(u_{(\ell-1)})+(\zeta-u_{(\ell-1)})\,S^w_\ell .
    \label{eq:affine-segment}
\end{equation}
Algorithm~\ref{alg:water-level} performs a single scan across these segments. It halts upon locating the initial index $\ell^\star$ that satisfies $W_\Gm(u_{(\ell^\star)})\ge c^\star$, subsequently outputting the precise root
\begin{equation}
    \zeta^\star\ =\ u_{(\ell^\star-1)}+\frac{c^\star-W_\Gm(u_{(\ell^\star-1)})}{S^w_{\ell^\star}} ,
    \label{eq:exact-zeta-root}
\end{equation}
The strictly positive nature of the costs ensures that $S^w_{\ell^\star}>0$. Finally, Algorithm~\ref{alg:grouped-allocation} integrates the two phases.

\input{algorithm/algorithm}

\begin{theorem}[Complexity]\label{th:complexity}
Algorithm~\ref{alg:grouped-allocation} computes the exact allocation $z^\gp$, requiring $\mathcal{O}(n+m\log m+|\Gm|\log|\Gm|)$ time along with $\mathcal{O}(m+|\Gm|)$ auxiliary space. Should one determine the level across $\Gm$ via median-based selection~\cite{dyer1984,zemel1984} rather than straightforward sorting, the runtime improves to $\mathcal{O}(n+m\log m)$. Furthermore, if the input already provides the representative ratios in sorted order, the execution time drops to $\mathcal{O}(n)$.
\end{theorem}

\begin{proof}
The algorithm calculates both allocation demands and the $m$ ratios within $\mathcal{O}(n)$ time. We then sort these $m$ ratios in $\mathcal{O}(m\log m)$ operations before identifying $k^\star$ via prefix sums, which takes $\mathcal{O}(m)$ time. Then the algorithm applies values of either $u_i$ or $0$ across the $n-|\Gm|$ elements located outside $\Gm$, demanding $\mathcal{O}(n)$ time. Sorting the ceilings found within $\Gm$ takes $\mathcal{O}(|\Gm|\log|\Gm|)$ steps, and from there the subsequent traversal in Algorithm~\ref{alg:water-level} runs in linear time. We can verify the scan's correctness: beginning iteration $\ell$, we know that $W=W_\Gm(u_{(\ell-1)})<c$ alongside $S^w=S^w_\ell$. Consequently, by~\eqref {eq:affine-segment}, checking whether $W+S^w\Delta u\ge c$ identifies the earliest segment where $W_\Gm$ meets or exceeds $c$. The function then yields the exact value shown in~\eqref{eq:exact-zeta-root}. This verification step must succeed no later than $\ell=q-1$, given that $W_\Gm(u_{(q-1)})=F^w_\Gm\ge c$. Regarding memory, the algorithm stores $|\Gm|$ sorted ceilings and $m$ prefix sums. Alternatively, identifying the target segment of a monotonically increasing piecewise-linear function with $|\Gm|$ breakpoints is achievable through weighted-median selection, which eliminates half of the remaining breakpoints in each pass and finishes in $\mathcal{O}(|\Gm|)$ time. This alternative method directly justifies our second complexity bound, while the final $\mathcal{O}(n)$ bound naturally emerges when Stage~1 no longer requires a sorting step.
\end{proof}

%% file: algorithm/algorithm.tex
\begin{figure*}[!t]
% \centering

\begin{minipage}[t]{0.48\textwidth}
\begin{algorithm}[H]
\caption{Grouped Knapsack Allocation}
\label{alg:grouped-allocation}
\SetKwInOut{Input}{Input}
\SetKwInOut{Output}{Output}
\Input{$\{v_i,w_i,u_i\}_{i\in[n]}$, $\Gc$ with $(\widehat{a}_k)_{k\in[m]}$, $C$.}
\Output{Allocation $z^\gp\in\Zc(I)$.}
\tcp{Stage 1: budget shares}
\For{$k\leftarrow0$ \KwTo $m-1$}{
    $\rhorep_k\leftarrow v(\widehat{a}_k)/w(\widehat{a}_k)$\;
    $F^w_{G_k}\leftarrow\sum_{i\in G_k}w_iu_i$\;
}
Sort the groups so that $\rhorep_0\ge\dots\ge\rhorep_{m-1}$\;
$F^w_{\le-1}\leftarrow0$\;
\lFor{$k\leftarrow0$ \KwTo $m-1$}{
    $F^w_{\le k}\leftarrow F^w_{\le k-1}+F^w_{G_k}$
}
$k^\star\leftarrow\min\{k\in[m]:F^w_{\le k}>C\}$\;
$c^\star\leftarrow C-F^w_{\le k^\star-1}$\;

\BlankLine
\tcp{Stage 2: Allocation}
\lFor{$j<k^\star$, $i\in G_j$}{
    $z^\gp_i\leftarrow u_i$
}
\lFor{$j>k^\star$, $i\in G_j$}{
    $z^\gp_i\leftarrow0$
}
$\Gm\leftarrow G_{k^\star}$\;
$\zeta^\star\leftarrow
\operatorname{SolveWaterLevel}
(\Gm,(w_i)_{i\in\Gm},(u_i)_{i\in\Gm},c^\star)$\;
\lForEach{$i\in\Gm$}{
    $z^\gp_i\leftarrow\min\{u_i,\zeta^\star\}$
}
\Return{$z^\gp$}\;
\end{algorithm}
\end{minipage}
\hfill
\begin{minipage}[t]{0.48\textwidth}
\begin{algorithm}[H]
\caption{\textsc{SolveWaterLevel}}
\label{alg:water-level}
\SetKwInOut{Input}{Input}
\SetKwInOut{Output}{Output}
\Input{$\Gm$, $\{w_i\}_{i\in\Gm}$, $\{u_i\}_{i\in\Gm}$, $c\in[0,F^w_\Gm]$.}
\Output{Level $\zeta^\star\in[0,\max_{i\in\Gm}u_i]$ with $W_\Gm(\zeta^\star)=c$.}

\BlankLine
\lIf{$c=0$}{\Return{$0$}}
\lIf{$c\ge F^w_\Gm$}{
    \Return{$\max_{i\in\Gm}u_i$}
}
$q\leftarrow|\Gm|$\;
sort $\Gm$ so that
$u_{(0)}\le\dots\le u_{(q-1)}$\;
$u_{(-1)}\leftarrow0$\;
$W\leftarrow0$\;
$S^w\leftarrow\sum_{r=0}^{q-1}w_{(r)}$\;

\For{$\ell\leftarrow0$ \KwTo $q-1$}{
    $\Delta u\leftarrow u_{(\ell)}-u_{(\ell-1)}$\;
    \lIf{$W+S^w\Delta u\ge c$}{
        \Return{$u_{(\ell-1)}+(c-W)/S^w$}
    }
    $W\leftarrow W+S^w\Delta u$\;
    $S^w\leftarrow S^w-w_{(\ell)}$\;
}
\Return{$u_{(q-1)}$}\tcp*{unreachable: $W_\Gm(u_{(q-1)})=F^w_\Gm>c$}
\end{algorithm}
\end{minipage}

\end{figure*}

%% file: utils/guarantees.tex
\section{Theoretical Guarantees}\label{sec:guarantees}
\noindent
We define a loss function that quantifies how the grouped allocation performs relative to the exact optimal solution:
\begin{equation}
    E(C)\ \coloneqq\ V^\star(C)-V^\gp(C)\ \ge\ 0 ,
    \label{eq:error-def}
\end{equation}
% This value remains non-negative since $z^\gp\in\Zc(I)$. We assume an order-compatible grouping (Definition~\ref{def:ordercomp}) in Sections~\ref{sec:localisation} and~\ref{sec:stability-sec}, whereas Section~\ref{sec:interleave} handles groups containing overlapping ratio intervals.

% \subsection{Intra-group loss}\label{sec:intra-price}
\noindent
Take a single group $G$ from a metric grouping having a share of $c\in[0,F^w_G]$. Suppose $z^\wf=z^\wf(c)$ denotes the allocation, while $z^G$ denotes the optimal allocation that assigns this same share among the items in $G$ using their actual parameters. We define the intra-group loss as
\begin{equation}
    \Delta_G(c)\ \coloneqq\ \sum_{i\in G}v_iz^G_i-\sum_{i\in G}v_iz^\wf_i\geq 0
    \label{eq:gap-def}
\end{equation}
% This quantity cannot be negative because $z^\wf$ provides a feasible solution to the subproblem. Since the macro stage determines the share $c$, it is necessary to establish a bound that applies uniformly across $c$.

\begin{theorem}[Intra-group loss]\label{th:inner}
For every group $G$ and every $c\in[0,F^w_G]$,
\begin{equation}
    \Delta_G(c)\ \le\ \vrep\,U_G\,\frac{\whi-\wlo}{\whi+\wlo}+\varepsilon_vU_G
    \ \le\ |G|\Bigl(\vrep\,\frac{\whi-\wlo}{\whi+\wlo}+\varepsilon_v\Bigr)\le |G|\Bigl(\frac{\vrep L_w}{\wmin}+L_v\Bigr)\delta\ =\ \mathcal{O}(\delta)
    \label{eq:inner-bound}
\end{equation}
% Given that $\whi-\wlo\le2L_w\delta$, $\whi+\wlo\ge2\wmin$ and $\varepsilon_v=L_v\delta$, we obtain
% \begin{equation}
%     \Delta_G(c)\ \le\ U_G\Bigl(\frac{\vrep L_w}{\wmin}+L_v\Bigr)\delta\ =\ \mathcal{O}(\delta).
%     \label{eq:inner-bound-delta}
% \end{equation}
\end{theorem}

\begin{proof}
Since $v_i=\vrep+(v_i-\vrep)$, we have
\[
\Delta_G(c)
=\sum_{i\in G}v_i(z_i^G-z_i^\wf)
=\vrep\sum_{i\in G}(z_i^G-z_i^\wf)
+\sum_{i\in G}(v_i-\vrep)(z_i^G-z_i^\wf)
\]
By \eqref{eq:eps-def},$|v_i-\vrep|\le \varepsilon_v$. Since $0\le z_i^G,z_i^\wf\le u_i$, we have $|z_i^G-z_i^\wf|\le u_i$, and therefore
\begin{gather}
\sum_{i\in G}(v_i-\vrep)(z_i^G-z_i^\wf)
\le
\varepsilon_v\sum_{i\in G}|z_i^G-z_i^\wf|
\le \varepsilon_vU_G
\end{gather}

Now we define $d_i=u_i-z_i\ge0$. Both allocations consume exactly \(c\): this holds for $z^\wf$ by \textbf{F2}, and for $z^G$ because every $v_i>0$ by Assumption~\ref{as:data}, so an optimal allocation of the share never leaves budget unspent. Hence $\sum_{i\in G}w_id_i=F_G^w-c$. Because, $\wlo\le w_i\le\whi$, we obtain
\begin{gather}
\sum_{i\in G}z_i^G\le\frac{c}{\wlo},
\
\sum_{i\in G}d_i^\wf\le\frac{F_G^w-c}{\wlo}
\ \because
c=\sum_{i\in G}w_iz_i^G\ge
\wlo\sum_{i\in G}z_i^G \nonumber \\
F_G^w-c=\sum_{i\in G}w_id_i^\wf\ge\wlo\sum_{i\in G}d_i^\wf
\end{gather}
Similarly,
\begin{gather}
\sum_{i\in G}z_i^\wf\ge\frac{c}{\whi},
\sum_{i\in G}d_i^G\ge\frac{F_G^w-c}{\whi}
\ \because
c=\sum_{i\in G}w_iz_i^\wf\le
\whi\sum_{i\in G}z_i^\wf\\
F_G^w-c=\sum_{i\in G}w_id_i^G
\le
\whi\sum_{i\in G}d_i^G.
\end{gather}

Now using $z_i=u_i-d_i,U_G=\sum_{i\in G}u_i$, we have 
\begin{align}
\sum_{i\in G}(z_i^G-z_i^\wf)
&=
\sum_{i\in G}z_i^G+\sum_{i\in G}d_i^\wf-U_G 
\le
\frac{c}{\wlo}
+\frac{F_G^w-c}{\wlo}
-U_G
=
\frac{F_G^w}{\wlo}-U_G.
\end{align}

Also,
\begin{align}
\sum_{i\in G}(z_i^G-z_i^\wf)
=
U_G-\sum_{i\in G}d_i^G-\sum_{i\in G}z_i^\wf
\le
U_G-\frac{F_G^w-c}{\whi}-\frac{c}{\whi}
=
U_G-\frac{F_G^w}{\whi}.
\end{align}

Taking the weighted average with weights $\frac{\wlo}{\wlo+\whi}$ and $\frac{\whi}{\wlo+\whi}$, we get
\begin{align}
\sum_{i\in G}(z_i^G-z_i^\wf)
\le
\frac{\wlo}{\wlo+\whi}
\left(\frac{F_G^w}{\wlo}-U_G\right)
+
\frac{\whi}{\wlo+\whi}
\left(U_G-\frac{F_G^w}{\whi}\right)
=
U_G\frac{\whi-\wlo}{\whi+\wlo}\\
\Rightarrow
\Delta_G(c)
\le
\vrep U_G\frac{\whi-\wlo}{\whi+\wlo}
+\varepsilon_vU_G=
U_G\left(
\vrep\frac{\whi-\wlo}{\whi+\wlo}
+\varepsilon_v
\right) 
\end{align}

Finally, since \(u_i\le1\), $U_G=\sum_{i\in G}u_i\le |G|$, Thus,
\begin{gather}
\Delta_G(c)
\le
|G|\left(
\vrep\frac{\whi-\wlo}{\whi+\wlo}
+\varepsilon_v
\right)
\end{gather}
The term $\frac{\whi-\wlo}{\whi+\wlo}$ represents the minimax relative error associated with the harmonic mean across the interval $[\wlo,\whi]$~\cite{porta1987}. This term represents the penalty incurred for overlooking cost variations within a specific group, and it persists even when all items share identical values (i.e., $\varepsilon_v=0$).
\end{proof}
% \subsection{Tightness of the harmonic factor}\label{sec:tight}
\begin{theorem}[Tightness]\label{th:tight}
For every $k\ge2$ and $\eta>0$, there exists a group $G$ of $k$ items with $u_i\in[0,1],v_i=\vrep$, and a share $c$ such that
\begin{gather}
\Delta_G(c)
\ge
(1-\eta)\vrep U_G
\frac{\whi-\wlo}{\whi+\wlo}.
% \tag{\ref{eq:tight-statement}}
\end{gather}
Hence, for groups of $k$ items, the supremum of the ratio between the actual loss and the bound in Theorem~\ref{th:inner} is $1$. If, in addition, all ceilings are equal to $1$ and all values are identical, the supremum is $1-\frac1k$.
\end{theorem}

\begin{proof}
Fix $r>1$ and first take $k=2$. Let $v=(1,1)$, $w=(1,r)$, $u=\left(1,\frac1r\right)$,$c=1$. Then $\wlo=1$, $\whi=r$, $U_G=1+\frac1r$, $\vrep=1$. Thus, $\rho_0=1$, $\rho_1=\frac1r\Rightarrow\rho_0>\rho_1\because r>1$. Therefore, the optimal allocation is,
\begin{gather}
z^G=(1,0),
\sum_i z_i^G=1\nonumber 
\end{gather}
For the water-filling allocation of Theorem~\ref{th:characterise}, the level $\zeta$ satisfies
\begin{gather}
\min\{1,\zeta\}
+r\min\left\{\frac1r,\zeta\right\}=1 \Rightarrow\zeta=\frac{1}{1+r}\because \frac{1}{1+r}<\frac{1}{r} \nonumber \\
\Rightarrow z^\wf=(\zeta,\zeta) \Rightarrow \sum_i z_i^\wf=\frac{2}{1+r} \nonumber \\
\Rightarrow\Delta_G(1)=1-\frac{2}{1+r}=\frac{r-1}{r+1} \nonumber \\
\Rightarrow\frac{\Delta_G(1)}
{U_G\frac{\whi-\wlo}{\whi+\wlo}}
=\frac{\frac{r-1}{r+1}}{(1+\frac{1}{r})\frac{r-1}{r+1}}=\frac{r}{r+1}
\end{gather}
Since $\frac{r}{r+1}\ge1-\eta$ as $r\ge\frac1\eta$, the ratio can be made arbitrarily close to $1$. For $k>2$, add $k-2$ items with $u_i=0$, $w_i=1$ and $v_i=\vrep=1$. These items receive zero allocation, and since $\wlo=1$ and $\whi=r$ are already attained by the first two items, $\Delta_G(c)$, $U_G$,$\wlo$,$\whi$ remain unchanged. Thus the same lower bound holds for every $k\ge2$.

% \paragraph{Unit ceilings: lower bound.}
Now let $u_i=1$, $v_i=\vrep,\forall i$. Take $k-1$ items with cost $1$ and one item with cost $r>1$, $w=(1,\ldots,1,r)$,$c=k-1$. The optimal allocation fills all $k-1$ "cheap" items: $\sum_i z_i^G=k-1$. For allocation,
\begin{gather}
(k-1)\zeta+r\zeta=k-1\Rightarrow
\zeta=\frac{k-1}{k-1+r}<1\ 
\because r>1 \\
\Rightarrow
\sum_i z_i^\wf=k\zeta=\frac{k(k-1)}{k-1+r}\nonumber\\
\Rightarrow
\frac{\Delta_G(c)}{\vrep}=
(k-1)-\frac{k(k-1)}{k-1+r}=\frac{(k-1)(r-1)}{k-1+r}
\end{gather}
Here,$U_G=k,\wlo=1,\whi=r$, so the bound is $\vrep k\frac{r-1}{r+1}$. Thus the ratio is $\frac{\Delta_G(c)}{\vrep k(r-1)/(r+1)}=\frac{k-1}{k}\frac{r+1}{k-1+r}$. Taking $r\to\infty$ gives $\lim_{r\to\infty}\frac{k-1}{k}\frac{r+1}{k-1+r}=1-\frac1k$. Therefore,
\[
\sup\frac{\Delta_G(c)}
{\vrep U_G\frac{\whi-\wlo}{\whi+\wlo}}
\ge
1-\frac1k.
\]
% \paragraph{Unit ceilings: upper bound.}
Now we assume, $\forall i,u_i=1$, $v_i=\vrep$. We sort the items so that $w_{(1)}\le\cdots\le w_{(k)}$. We now define $S_j=\sum_{i=1}^j w_{(i)}$,$F=S_k=\sum_{i=1}^k w_{(i)}$. For $0\le\zeta\le1$,$W_G(\zeta)=F\zeta$ as $u_i=1$. Thus
$
\zeta=\frac{c}{F}$, and
$
\sum_i z_i^\wf
=
k\frac{c}{F}
$. Let,
$
M(c)=\sum_i z_i^G.
$.
By the greedy rule, items are filled in increasing order of cost. Hence $\forall j\in[0,k],M(S_j)=j$. Moreover, $M(c)$ is linear between consecutive points $S_j$ and $S_{j+1}$. Therefore,
$
\Delta_G(c)
=
\vrep\left(M(c)-\frac{kc}{F}\right)
$ is piecewise linear, so its maximum is attained at some breakpoint $c=S_j$. For $j=0$ and $j=k$, $\Delta_G(S_j)=0$. Thus consider $1\le j\le k-1$. We have
\[
\frac{\Delta_G(S_j)}{\vrep}
=
j-\frac{kS_j}{F}
=
j-\frac{kS_j}{S_j+(F-S_j)}.
\]
Since $S_j\ge j\wlo$, and $F-S_j\le(k-j)\whi$, we obtain
\begin{gather}
\frac{S_j}{S_j+(F-S_j)}
\ge
\frac{j\wlo}{j\wlo+(k-j)\whi},
\because
\frac{x}{x+y}
\text{ increases with }x
\text{ and decreases with }y \nonumber \\
\Rightarrow
\frac{\Delta_G(S_j)}{\vrep}
\le
j-\frac{kj\wlo}
{j\wlo+(k-j)\whi}=
\frac{j(k-j)(\whi-\wlo)}
{j\wlo+(k-j)\whi} \nonumber \\
\Rightarrow
(k-1)\bigl(j\wlo+(k-j)\whi\bigr)
-j(k-j)(\wlo+\whi)=
j(j-1)\wlo
+(k-j)(k-j-1)\whi
\ge0\nonumber \\
\because
1\le j\le k-1,
\wlo>0,
\whi>0 \nonumber \\
\Rightarrow
\frac{j(k-j)}
{j\wlo+(k-j)\whi}
\le
\frac{k-1}{\wlo+\whi} \nonumber \\
\Rightarrow
\frac{\Delta_G(S_j)}{\vrep}
\le
(k-1)
\frac{\whi-\wlo}{\whi+\wlo}.
\end{gather}
Since $U_G=k$ as $u_i=1$, we conclude that
\begin{gather}
\Delta_G(c)
\le
\left(1-\frac1k\right)
\vrep U_G
\frac{\whi-\wlo}{\whi+\wlo} \nonumber \\
\Rightarrow
\sup
\frac{\Delta_G(c)}
{\vrep U_G(\whi-\wlo)/(\whi+\wlo)}
=
1-\frac1k
\end{gather}
For general ceilings, the first construction gives a ratio arbitrarily close to $1$, and therefore $\boxed{\sup=1}$.
\end{proof}

% \noindent
% The setup involving unit ceilings is illustrated in Figure~\ref{fig:gap}. Here, the allocation mechanism forces the costly outlier up to the shared level $\zeta$, thereby depriving the less expensive items of funding. In Theorem~\ref{th:inner}, the local denominator term $\whi+\wlo$ plays a crucial role. If one were to substitute it with the global term $2\wmin$, the asymptotic order shown in~\eqref{eq:inner-bound-delta} remains unchanged. However, for the specific family described above where $\wmin=1$, this substitution would evaluate to $k(r-1)/2=\Theta(r)$. This creates a contradiction, as the total loss is strictly bounded by the group capacity, which is only $k-1$.

% \input{diagrams/water_fill}

% \subsection{Boundary localisation}\label{sec:localisation}
\noindent
If we sum the result of Theorem~\ref{th:inner} across every group, the overall error $E(C)$ is bounded by $\mathcal{O}(n\delta)$. However, this estimate proves far too pessimistic. The precise optimal solution completely allocated to any item situated above the clearing ratio and provides nothing to those below it. Therefore, provided the grouping aligns with the true ratio ordering, the only possible discrepancies between the two allocations occur within the single group where the budget is ultimately exhausted.

\begin{definition}[Order-compatibility]\label{def:ordercomp}
We call a metric grouping $\Gc$, which is indexed to satisfy $\rhorep_0\ge\dots\ge\rhorep_{m-1}$, \emph{order-compatible} whenever $\min_{i\in G_k}\rho_i\ge\max_{l\in G_{k'}}\rho_l,\forall k<k'$ meaning that the ratio intervals corresponding to the groups exhibit disjoint interiors and strictly decrease as the group index increases.
\end{definition}

\begin{lemma}[Resolution margin]\label{lem:margin}
Given Assumption~\ref{as:data}, the function mapping $\rho(a)\coloneqq \frac{v(a)}{w(a)}$ exhibits Lipschitz continuity characterized by the constant
\begin{equation}
    L_\rho\ \le\ \frac{L_v}{\wmin}+\frac{v_{\max}L_w}{\wmin^2},
    \label{eq:lip-rho}
\end{equation}
thereby ensuring that $\rho_i\in[\rhorep_k-\delta L_\rho,\ \rhorep_k+\delta L_\rho]$ holds true for all $i\in G_k$. Should the condition $\rhorep_k-\rhorep_{k+1}\ >\ 2\delta L_\rho\forall k\in[m-1]$, be met, the grouping $\Gc$ becomes order-compatible.
\end{lemma}

\begin{proof}
Taking any $a,a'\in A$, we split $\rho(a)-\rho(a')=\frac{v(a)-v(a')}{w(a)}+v(a')\frac{w(a')-w(a)}{w(a)w(a')}$, whence
\[
    |\rho(a)-\rho(a')|\ \le\ \frac{|v(a)-v(a')|}{w(a)}+v(a')\,\frac{|w(a)-w(a')|}{w(a)\,w(a')}\ \le\ \Bigl(\frac{L_v}{\wmin}+\frac{v_{\max}L_w}{\wmin^2}\Bigr)d(a,a'),
\]
using $w(a),w(a')\ge\wmin$ and $v(a')\le v_{\max}$. This directly yields~\eqref{eq:lip-rho}.

\noindent
Applying this to $a=a_i$ and $a'=\widehat a_k$ for $i\in G_k$, Definition~\ref{def:grouping} gives $d(a_i,\widehat a_k)\le\delta$, so $|\rho_i-\rhorep_k|\le\delta L_\rho$, i.e. $\rho_i\in[\rhorep_k-\delta L_\rho,\rhorep_k+\delta L_\rho]$.

\noindent
Assume now $\rhorep_k-\rhorep_{k+1}>2\delta L_\rho$ for every $k\in[m-1]$. Fix $k<k'$. Because the groups are indexed so that $\rhorep_0\ge\dots\ge\rhorep_{m-1}$, telescoping over the consecutive gaps gives $\rhorep_k-\rhorep_{k'}\ \ge\ \rhorep_k-\rhorep_{k+1}\ >\ 2\delta L_\rho$. Hence, for any $i\in G_k$ and $l\in G_{k'}$,
\[
    \rho_i\ \ge\ \rhorep_k-\delta L_\rho\ >\ \rhorep_{k'}+\delta L_\rho\ \ge\ \rho_l .
\]
Taking the minimum over $i\in G_k$ and the maximum over $l\in G_{k'}$ yields $\min_{i\in G_k}\rho_i\ge\max_{l\in G_{k'}}\rho_l$ for all $k<k'$, which is precisely order-compatibility (Definition~\ref{def:ordercomp}).
\end{proof}

\begin{theorem}[Boundary localisation]\label{th:localise-master}
Suppose $\Gc$ is order-compatible and $C<F^w_{[n]}$. Then~\eqref{eq:P} admits an optimal solution $z^\star$ satisfying $z^\star_i=u_i$ for items in $\Gm^-$, $z^\star_i=0$ for those in $\Gm^+$, and $\sum_{i\in\Gm}w_iz^\star_i=c^\star$. As a direct consequence, the allocations $z^\star$ and $z^\gp$ are identical outside of $\Gm$, leaving us with
\begin{equation}
    E(C)\ =\ \Delta_\Gm(c^\star)\ \le\ U_\Gm\Bigl(\vrep_\Gm\,\frac{\whi_\Gm-\wlo_\Gm}{\whi_\Gm+\wlo_\Gm}+\varepsilon_v\Bigr).
    \label{eq:localised-error}
\end{equation}
Whenever $|G_k|\le K$ is satisfied for every $k$, it follows that
\begin{equation}
    \frac{E(C)}{n}\ \le\ \frac{K}{n}\Bigl(\frac{\vrep_\Gm L_w}{\wmin}+L_v\Bigr)\delta\ =\ \mathcal{O}\Bigl(\frac{K}{n}\Bigr),
    \label{eq:asymptotic-vanish}
\end{equation}
This bound applies uniformly across any family of instances where $\vrep_\Gm\le\bar V$ and $\wmin\ge\bar w>0$. Furthermore, the relative loss is constrained by
\begin{equation}
    \frac{E(C)}{V^\star(C)}\ \le\ \frac{U_\Gm}{C}\cdot\frac{\vrep_\Gm}{\nu^\star}\Bigl(\frac{\whi_\Gm-\wlo_\Gm}{\whi_\Gm+\wlo_\Gm}+\frac{\varepsilon_v}{\vrep_\Gm}\Bigr).
    \label{eq:relative-bound}
\end{equation}
\end{theorem}
\begin{proof}
We order the items by non-increasing density $\rho_i:=\frac{v_i}{w_i}$, breaking ties by prioritising the smaller group index. By the order-compatibility assumption, all items in $\mathcal G^-$ precede all items in $\mathcal G$, which in turn precede all items in $\mathcal G^+$. Hence, when Dantzig's rule is applied to this ordering, every item in $\mathcal G^-$ is fully saturated before any item in $\mathcal G^+$ can receive positive allocation. Since $C<F^w_{[n]}$, the budget constraint is binding. Moreover, by the definition of $k^\star$ and the fact that $F^w_{\le k^\star-1}\le C<F^w_{\le k^\star}$, the total cost required to saturate all items preceding the boundary group is $F^w_{\le k^\star-1}=F^w_{\mathcal G^-}\le C$. Consequently, after saturating $\mathcal G^-$, the residual budget is
\[
    c^\star:=C-F^w_{\mathcal G^-},\text{such that } 0\le c^\star<F^w_{\mathcal G}
\]
Dantzig's rule therefore allocates the residual amount $c^\star$ entirely among the items of $\mathcal G$ and assigns zero allocation to every item in $\mathcal G^+$. Thus, the resulting allocation satisfies
\[
    z_i^\star=\begin{cases}
        \forall i\in\mathcal G^-\ ,u_i,\\
        \forall i\in\mathcal G^+\ ,z_i^\star=0,\\
        \forall i\in\mathcal G\ ,\sum_{i\in\mathcal G}w_i z_i^\star=c^\star
    \end{cases}
\]
By Lemma~\ref{lem:price}, Dantzig's rule produces an optimal solution of \eqref{eq:P}. Hence this particular allocation may be taken as an optimal solution $z^\star$. It remains to compare $z^\star$ with the groupwise approximation $z^{\gp}$. Both allocations coincide on $\mathcal G^-$, where all items are saturated, and on $\mathcal G^+$, where all items receive zero allocation. On $\mathcal G$, both allocations solve the same residual-budget problem with budget $c^\star$. Therefore,
\begin{align}
    V^\star(C)&=\sum_{i\in\mathcal G^-}v_i u_i+\sum_{i\in\mathcal G}v_i z_i^\star,\ 
    V^{\gp}(C)=\sum_{i\in\mathcal G^-}v_i u_i+\sum_{i\in\mathcal G}v_i z_i^{\wf}(c^\star)\ \nonumber \\
    \Rightarrow E(C)&=V^\star(C)-V^{\gp}(C)=\Delta_{\mathcal G}(c^\star)\nonumber \\
    &\le
    U_{\mathcal G}
    \left(
        \vrep_{\mathcal G}
        \frac{\whi_{\mathcal G}-\wlo_{\mathcal G}}
             {\whi_{\mathcal G}+\wlo_{\mathcal G}}
        +\varepsilon_v
    \right)\because 0\le c^\star<F^w_{\mathcal G}\nonumber \\
    &\le K \left(
        \frac{\vrep_{\mathcal G}L_w}{\wmin}
        +L_v
    \right)\delta\ \because
    U_{\mathcal G}\le |\mathcal G|\le K\ \&\ 
    \frac{\whi_{\mathcal G}-\wlo_{\mathcal G}}{\whi_{\mathcal G}+\wlo_{\mathcal G}}
    \le
    \frac{\whi_{\mathcal G}-\wlo_{\mathcal G}}
         {2\wmin}
    \le
    \frac{L_w\delta}{\wmin},\ 
    \varepsilon_v\le L_v\delta \nonumber \\
\Rightarrow
 \frac{E(C)}{n}&\le\frac{K}{n}\left(\frac{\vrep_{\mathcal G}L_w}{\wmin}+L_v\right)\delta=\mathcal O\left(\frac Kn\right)
\end{align}
which proves \eqref{eq:localised-error}. Provided the quantities $\vrep_{\mathcal G}$ and $\wmin^{-1}$ are uniformly bounded. In particular, the bound is uniform over any family of instances satisfying $\vrep_{\mathcal G}\le\bar V$, $\wmin\ge\bar w>0$. Finally, Lemma~\ref{lem:price} gives the lower bound $V^\star(C)\ge \nu^\star C$. Combining this with \eqref{eq:localised-error} yields
\begin{align}
    \frac{E(C)}{V^\star(C)}
    \le
    \frac{U_{\mathcal G}}{\nu^\star C}
    \left(
        \vrep_{\mathcal G}
        \frac{\whi_{\mathcal G}-\wlo_{\mathcal G}}
             {\whi_{\mathcal G}+\wlo_{\mathcal G}}
        +\varepsilon_v
    \right)
    =
    \frac{U_{\mathcal G}}{C}
    \frac{\vrep_{\mathcal G}}{\nu^\star}
    \left(
        \frac{\whi_{\mathcal G}-\wlo_{\mathcal G}}
             {\whi_{\mathcal G}+\wlo_{\mathcal G}}
        +
        \frac{\varepsilon_v}{\vrep_{\mathcal G}}
    \right)\nonumber 
\end{align}
which is precisely \eqref{eq:relative-bound}. If the budget hits an allocation saturation threshold $C=F^w_{\le k}$, the remaining share becomes exactly $c^\star=0$. Under these conditions, both allocation methods fully fund groups $G_0,\dots, G_k$ and leave everything else empty, culminating in an error of $E(C)=0$.
\end{proof}

\begin{proposition}[]\label{prop:graded-overlap}
Consider any arbitrary metric grouping denoted by $\Gc$, sorted by descending representative ratios. Let us define $\omega\ \coloneqq\ \max_{k<k'}\Bigl(\max_{l\in G_{k'}}\rho_l-\min_{i\in G_k}\rho_i\Bigr)^{+}
    \label{eq:overlap-def}$
% \end{equation}
to represent the maximum possible ratio inversion occurring across the groups. Under this definition, for any budget $C<F^w_{[n]}$, the error is constrained by
\begin{equation}
    E(C)\ \le\ \Delta_\Gm(c^\star)+\omega C ,
    \label{eq:graded-bound}
\end{equation}
where Theorem~\ref{th:inner} provides the bound for $\Delta_\Gm(c^\star)$. If we set $\omega=0$, the grouping achieves order compatibility, and the expression in~\eqref{eq:graded-bound} matches the bound established in Theorem~\ref{th:localise-master}.
\end{proposition}

\begin{proof}
We construct an allocation $\tilde z$ that matches $z^\gp$ outside $\Gm$ and acts as an optimal distribution of the share $c^\star$ strictly within $\Gm$. This construction gives $\sum_iv_i(\tilde z_i-z^\gp_i)=\Delta_\Gm(c^\star)$. Consequently, we only need to demonstrate that $V^\star(C)-\sum_iv_i\tilde z_i\le\omega C$. Suppose $z^\star$ represents an optimal allocation characterised by the multiplier $\nu^\star$, consistent with Lemma~\ref{lem:price}. Because both $z^\star$ and $\tilde z$ exhaust the exact budget $C$, we find
\[
    V^\star(C)-\sum_iv_i\tilde z_i\ =\ \sum_i(v_i-\nu^\star w_i)(z^\star_i-\tilde z_i)\ =\ \sum_iw_i(\rho_i-\nu^\star)(z^\star_i-\tilde z_i)=T .
\]
Let us introduce the sets $P\coloneqq\{i:z^\star_i>\tilde z_i\}$ and $N\coloneqq\{i:z^\star_i<\tilde z_i\}$. According to condition~\eqref{eq:bang-bang}, the fact that $z^\star_i>0$ requires $\rho_i\ge\nu^\star$ for every item $i\in P$. Likewise, having $z^\star_i<u_i$ implies $\rho_i\le\nu^\star$ for any $i\in N$. Thus, every individual term contributing to $T$ must be non-negative. Because both allocations spend identical amounts, we deduce that $\Sigma\coloneqq\sum_{i\in P}w_i(z^\star_i-\tilde z_i)=\sum_{i\in N}w_i(\tilde z_i-z^\star_i)$. Furthermore, this total cannot exceed the budget, so $\Sigma\le\sum_iw_iz^\star_i=C$. In the event that $P=\emptyset$, we get $\Sigma=0$, which implies $N=\emptyset$ and forces $T=0$. In all other cases, we have
\begin{equation}
    T\ \le\ \Sigma\Bigl[\max_{i\in P}(\rho_i-\nu^\star)+\max_{j\in N}(\nu^\star-\rho_j)\Bigr]\ \le\ C\max_{i\in P,\,j\in N}(\rho_i-\rho_j).
    \label{eq:T-bound}
\end{equation}
We pick any elements $i\in P$ and $j\in N$, assuming they belong to groups $G_k$ and $G_{k'}$ respectively. The inequality $\tilde z_i<z^\star_i\le u_i$ indicates that item $i$ does not belong to $\Gm^-$, thereby enforcing $k\ge k^\star$. Conversely, the relation $\tilde z_j>z^\star_j\ge0$ means item $j$ falls outside $\Gm^+$, leading to $k'\le k^\star$. Whenever $k>k'$, the difference satisfies $\rho_i-\rho_j\le\max_{l\in G_k}\rho_l-\min_{l\in G_{k'}}\rho_l\le\omega$, relying on~\eqref{eq:overlap-def}. If it happens that $k=k'=k^\star$, then both items are located within $\Gm$. Here, $\tilde z$ acts as the optimal allocation for the share $c^\star<F^w_\Gm$. Invoking~\eqref{eq:bang-bang} for this specific subproblem alongside its unique multiplier $\nu_\Gm$, the condition $\tilde z_i<u_i$ demands $\rho_i\le\nu_\Gm$, whereas $\tilde z_j>0$ forces $\rho_j\ge\nu_\Gm$. Together, these inequalities guarantee that $\rho_i-\rho_j\le0\le\omega$. Ultimately,~\eqref{eq:T-bound} concludes the proof by confirming $T\le\omega C$.
\end{proof}

% \subsection{Lipschitz stability}\label{sec:stability-sec}
\noindent
As Proposition~\ref{prop:eps} shows, the exact solution mapping lacks any finite Lipschitz modulus. However, the grouped allocation is Lipschitz continuous, provided that any parameter perturbation remains strictly within the separation margin between adjacent groups. For the remainder of this section, we assume the grouping $\Gc$ and its corresponding representatives remain constant, while individual item parameters may vary.

\begin{theorem}[Stability under a single-cost perturbation]
\label{th:stability}
Suppose the grouping $\mathcal G$ is order-compatible and has a positive margin
$
    \gamma
    \coloneqq
    \min_{k\in[m-1]}
    \bigl(\rhorep_k-\rhorep_{k+1}\bigr)
    -2\delta L_\rho
    >0
$.
Assume that $|G_k|\le K$ for every $k$, and that the budget has positive slack
$
    \beta
    \coloneqq
    \min_{k\in[m]}
    \bigl|C-F^w_{\le k}\bigr|
    >0
$.
Let the cost of an arbitrary item $j\in[n]$ be perturbed according to
\begin{gather}
    w'=w+\eta e_j\text{ where }
    |\eta|<
    \eta_0
    \coloneqq
    \min\left\{
        \frac{\gamma\wmin^2}
        {v_{\max}+\gamma\wmin},
        \beta
    \right\}.
\end{gather}
Then:
\begin{enumerate}[label=(\roman*)]
    \item the perturbed instance remains order-compatible with respect to $\mathcal G$; \label{th:stability_one}
    \item the boundary group $\mathcal G_m$ and the sets $\mathcal G^-$ and $\mathcal G^+$ remain unchanged. Consequently, $i\notin\mathcal G_m$, $z_i^{\gp}(w')=z_i^{\gp}(w)$; \label{th:stability_two}
    \item the groupwise allocation on the boundary group is Lipschitz in $\eta$:
        \[
            \bigl\|z^{\gp}(w')-z^{\gp}(w)\bigr\|_1
            \le
            \frac{|\mathcal G_m|}{\wmin}|\eta|
            \le
            \frac{K}{\wmin}|\eta|.
        \]\label{th:stability_three}
\end{enumerate}
In particular, the perturbation quotient is uniformly bounded by $\frac{K}{\wmin}$.
\end{theorem}

\begin{proof}
Since
$\rho_i=\frac{v_i}{w_i}$ and $\rho_j'=\frac{v_j}{w_j+\eta}$, 
we have
\begin{gather}
    |\rho_j'-\rho_j|
    =
    \frac{v_j|\eta|}
         {w_j|w_j+\eta|}
    \le
    \frac{v_{\max}|\eta|}
         {\wmin(\wmin-|\eta|)}.
\end{gather}
The definition of $\eta_0$ implies $|\eta|<\wmin$ and hence
$
    |\rho_j'-\rho_j|<\gamma.
$.
Suppose $j\in G_k$. By the margin assumption,
$
    \rho_j\ge \rhorep_k-\delta L_\rho
$.
Therefore,
\begin{gather}
    \rho_j'
    >
    \rhorep_k-\delta L_\rho-\gamma
    \ge
    \rhorep_{k+1}+\delta L_\rho
\end{gather}
By Lemma~\ref{lem:margin}, every item in $G_{k+1}$ has density at most
$\rhorep_{k+1}+\delta L_\rho$. Hence
$
    \rho_j'>
    \max_{i\in G_{k+1}}\rho_i
$.
The same argument, applied to the preceding group, gives
$
    \rho_j'<
    \min_{i\in G_{k-1}}\rho_i
$,
whenever $k\ge1$. The bound extends from the adjacent groups to all of them: for any $k'>k+1$ the ordering $\rhorep_{k'}\le\rhorep_{k+1}$ gives $\max_{l\in G_{k'}}\rho_l\le\rhorep_{k'}+\delta L_\rho\le\rhorep_{k+1}+\delta L_\rho<\rho_j'$, and symmetrically $\rho_j'<\min_{l\in G_{k'}}\rho_l$ for every $k'<k-1$. Thus the perturbed density remains inside the band separating its own group from every other one. Since all other densities are unchanged and already satisfy Definition~\ref{def:ordercomp}, every pair of groups still has the required ordering. Hence the perturbed instance
remains order-compatible with respect to $\mathcal G$. \boxed{Proved \ref{th:stability_one}}

\noindent
Now, for every prefix $G_1\cup\cdots\cup G_k$, the perturbation changes its
total cost by either $0$ or $\eta u_j$. Consequently,
$
    \bigl|F^{w'}_{\le k}-F^w_{\le k}\bigr|
    \le |\eta|
    <\beta.
$.
By definition of $\beta$,
$
    |F^w_{\le k}-C|\ge\beta,\forall k
$
Hence the perturbation is too small to move any prefix
cost across the budget $C$. The sign of
$
    F^w_{\le k}-C
$
is therefore the same before and after the perturbation. It follows that the index $k^\star$ satisfying
$
    F^w_{\le k^\star-1}\le C<F^w_{\le k^\star}
$
is unchanged. Thus the boundary group $\mathcal G_m=G_{k^\star}$,
together with $\mathcal G^-$ and $\mathcal G^+$, is unchanged. By Theorem~\ref{th:localise-master}, all items in $\mathcal G^-$ remain
fully allocated and all items in $\mathcal G^+$ remain unallocated.
Therefore $i\notin\mathcal G_m$
$
    z_i^{\gp}(w')=z_i^{\gp}(w)
$ \boxed{Proved \ref{th:stability_two}}

Let $c^\star=C-F^w_{\le k^\star-1}$,$c^{\star\prime}=C-F^{w'}_{\le k^\star-1}$. By part \ref{th:stability_two}, the boundary group does not change. The groupwise
allocation on $\mathcal G_m$ is obtained from an allocation budget. We define
$
    W(\xi;w)
    \coloneqq
    \sum_{i\in\mathcal G_m}
    w_i\min\{u_i,\xi\}.
$
Let $\zeta$ and $\zeta'$ denote the corresponding allocations, so that
$W(\zeta;w)=c^\star$, and $W(\zeta';w')=c^{\star\prime}$.
Since only item $j$ is perturbed, we have
\[
    |W(\zeta';w)-W(\zeta';w')|
    \le |\eta|
\]
Also,$|c^{\star\prime}-c^\star|\le |\eta|$. However, these two changes cannot occur simultaneously: the first is
non-zero only when $j\in\mathcal G_m$, while the second is non-zero only
when $j\in\mathcal G^-$. Hence
\[
    |W(\zeta';w)-c^\star|
    \le |\eta|.
\]
Let $i^\dagger\in\mathcal G_m$ be an item with
$
    u_{i^\dagger}
    =
    \max_{i\in\mathcal G_m}u_i
$.
For every $\xi$ between the two allocations, this item is not yet saturated, except possibly at the endpoint. Therefore $W(\xi;w)$ has slope at least
$
    w_{i^\dagger}\ge\wmin
$
between $\zeta$ and $\zeta'$. Consequently,
$
    |\zeta'-\zeta|
    \le
    \frac{|\eta|}{\wmin}
$. Finally,
$
    z_i^{\gp}(w)=\min\{u_i,\zeta\},
$
and the map $\xi\mapsto\min{u_i,\xi}$ is $1$-Lipschitz. Hence
$$
\begin{aligned}
    \bigl\|z^{\gp}(w')-z^{\gp}(w)\bigr\|_1
    &=
    \sum_{i\in\mathcal G_m}
    \left|
        \min\{u_i,\zeta'\}
        -
        \min\{u_i,\zeta\}
    \right|\\
    &\le
    |\mathcal G_m|\,|\zeta'-\zeta|\le
    \frac{|\mathcal G_m|}{\wmin}|\eta|\le
    \frac{K}{\wmin}|\eta|
\end{aligned}
$$

This proves the claimed \boxed{\ref{th:stability_three}}.
\end{proof}

%% file: utils/related_works.tex
\section{Related Work}\label{sec:related}
\noindent
Dantzig~\cite{dantzig1957} introduced the fractional knapsack problem alongside its greedy solution. Bounded formulations and other knapsack variants have been extensively explored~\cite{martello1990,pisinger2000,chen2024bounded}. Furthermore, previous works have examined bounded-variable programs of the form~\eqref {eq:P}~\cite{bitran1981,ibaraki1988,patriksson2008}. Within polyhedral sensitivity analysis, the discontinuous nature of optimal-solution maps in parametric linear programming remains a foundational topic~\cite{hoffman1952,robinson1980,clarke1990,rockafellar1998,bonnans2000,dontchev2009}. Commonly known as Sprumont's uniform rule~\cite{sprumont1991} and recognised as the constrained equal awards rule within axiomatic rationing~\cite{young1994,moulin2002,thomson2003}, water-filling provides the foundation for both $\alpha$-fair rate allocation and proportional networking approaches~\cite{mo2000,lan2010}. Our Theorem~\ref{th:characterise} extends these concepts to account for capacity ceilings along with heterogeneous costs. As detailed in Theorem~\ref{th:inner}, the resulting loss resembles a price of fairness as discussed in earlier studies~\cite{bertsimas2011,caragiannis2012,barman2020,nicosia2017}. However, this limitation arises from limited information rather than a direct fairness constraint. Consequently, this links our findings to mathematical programming aggregation bounds~\cite{geoffrion1977,zipkin1980vars,rogers1991}. Similar mechanisms for discretisation, metric clustering, and coarsening can be found throughout existing literature~\cite{mcafee2002,hoppe2010,dwork2012,chierichetti2017,jung2020,chakrabarti2022,gonzalez1985,hochbaum1985,lloyd1982,gray1998}. While robust optimisation accounts for parameter uncertainty~\cite{bental1999,bertsimas2004,bental2009,bertsimas2011theory,delage2010,blanchet2019}, it continues to sort based on nominal values. On the other hand, smoothing techniques abandon the solution's polyhedral sparsity~\cite{tikhonov1977,nesterov2005}. In contrast, grouped allocation preserves both the feasible region and the objective, solely restricting the internal rule applied to each group.

%% file: utils/Evaluation.tex
\section{Evaluation}\label{sec:evaluation}
\noindent
The preceding sections make three quantitative claims about the grouped allocation: the loss with respect to the exact optimum is confined to the boundary group and scales with $\delta$ (Theorems~\ref{th:inner} and~\ref{th:localise-master}); the relative loss depends on $\delta$ rather than on $n$ (Equation~\eqref{eq:relative-bound}), with the per-item bound~\eqref{eq:asymptotic-vanish} decaying as $\mathcal{O}(K/n)$ whenever the group-size cap $K$ stays bounded; and the two-stage procedure runs in $\mathcal{O}(n+m\log m+|\Gm|\log|\Gm|)$ time (Theorem~\ref{th:complexity}), which should make it competitive with Dantzig's $\mathcal{O}(n\log n)$ rule in practice. This section tests each claim empirically. We compare the grouped allocation $z^\gp$ against the exact greedy optimum $z^\star$ on randomly generated instances, measuring the relative loss $E(C)/V^\star(C)$ as a function of $n$ (Figure~\ref{fig:error_rate_vs_n}) and of $\delta$ (Figure~\ref{fig:error_rate_tolerance}), and the wall-clock runtime of both methods as a function of $n$ (Figure~\ref{fig:runtime_comparison}).

\subsection{Experimental setup}\label{sec:eval-setup}
\noindent
We implemented both algorithms in C++17 and compiled them with optimisations enabled. We ran all experiments single-threaded on an Intel~i7 processor. % CHECK: add compiler/flags (e.g., g++ -O2), exact CPU model, and RAM if available.
For the baseline, we use Dantzig's greedy rule~\cite{dantzig1957} as described in the proof of Lemma~\ref{lem:price}. Since this procedure returns an optimal solution to~\eqref {eq:P}, its value serves as the reference $V^\star(C)$ against which the loss $E(C)=V^\star(C)-V^\gp(C)$ in~\eqref {eq:error-def} is computed. For the grouped method, we run Algorithm~\ref{alg:grouped-allocation}.

% \paragraph{Instance generation.}
For each value of $n$, we draw $v_i$, $w_i$, and capacity limits $u_i\in[0,1]$ independently at random. % CHECK: state the distributions and ranges (e.g., $v_i\sim U(\cdot,\cdot)$, $w_i\sim U(\cdot,\cdot)$, $u_i\sim U(0,1)$) and the random seed / number of repetitions per data point.
The budget $C$ is set to a fixed fraction of the total demand $F^w_{[n]}$, so that $C<F^w_{[n]}$ holds and the budget is binding. % CHECK: state the fraction used.
The metric grouping of Definition~\ref{def:grouping} is built directly from the item attributes: the attribute space is partitioned into cells of radius $\delta$ around a fixed set of representatives, and every item is assigned to the cell containing it. Consequently, the number of groups $m$ is governed by the covering number of the attribute space at scale $\delta$ and does not grow with $n$ once $n$ is large. For the runtime experiments, the time reported for the grouped method includes this grouping step, so the comparison is against the full end-to-end cost of the procedure. % CHECK: confirm that the grouping step is included in the measured runtime, or adjust the sentence.

% \paragraph{Metrics.}
We report the relative loss $100\cdot E(C)/V^\star(C)$. Runtimes are reported in milliseconds of wall-clock time. Each data point is the average over several independently generated instances. % CHECK: state the number of repetitions.

\begin{figure}[!t]
    \centering
    \begin{subfigure}{0.49\columnwidth}
        \centering
        \includegraphics[width=\linewidth]{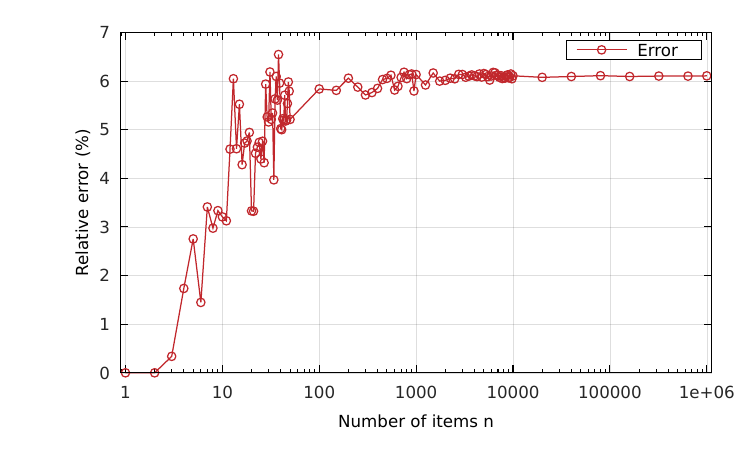}
        \caption{Relative loss versus $n$ at a fixed tolerance $\delta$.}
        \label{fig:error_rate_vs_n}
    \end{subfigure}
    \hfill
    \begin{subfigure}{0.49\columnwidth}
        \centering
        \includegraphics[width=\linewidth]{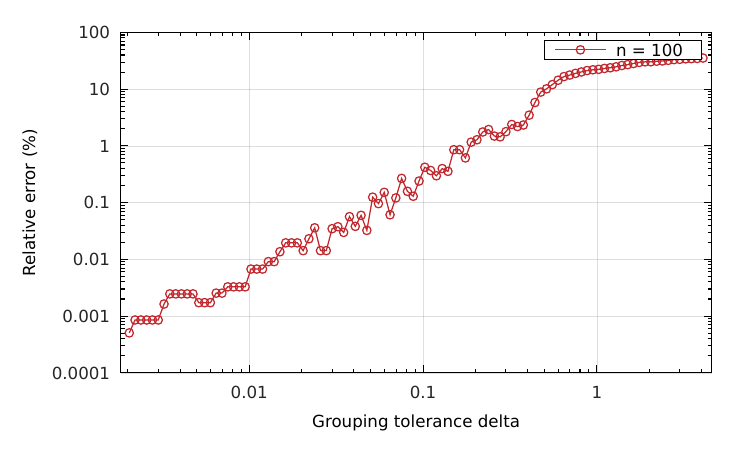}
        \caption{Relative loss versus $\delta$ at $n=100$ (log--log axes).}
        \label{fig:error_rate_tolerance}
    \end{subfigure}
    \caption{Relative loss $100\frac{E(C)}{V^\star(C)}$ of the grouped allocation with respect to the exact optimum. (a)~For a fixed $\delta$, the loss vanishes for tiny instances, where every group is a singleton, and settles at a constant level as $n$ grows, as predicted by~\eqref{eq:relative-bound}. (b)~For a fixed $n$ the loss grows polynomially with $\delta$ until a single group covers the whole instance, after which it saturates.}
    \label{fig:error_figures}
\end{figure}

\subsection{Loss as a function of the instance size}\label{sec:eval-n}
\noindent
Figure~\ref{fig:error_rate_vs_n} plots the relative loss against $n$, for a fixed $\delta$. % CHECK: state the value of $\delta$ used for this figure.
Three regimes are visible. For $n\le2$ the loss is exactly zero. With so few items, every group is a singleton, so the boundary group $\Gm$ contains a single item, the allocation of Theorem~\ref{th:characterise} reduces to $z^\wf_i=\frac{c^\star}{w_i}$, and the intra-group loss $\Delta_\Gm(c^\star)$ of~\eqref{eq:gap-def} vanishes. Theorem~\ref{th:localise-master} then gives $E(C)=\Delta_\Gm(c^\star)=0$, and the grouped allocation coincides with Dantzig's solution.

\noindent
For $n$ between roughly $3$ and $50$, the loss rises steeply and fluctuates strongly from one value of $n$ to the next. In this range, the groups begin to acquire more than one member, so the boundary group can now contain items with distinct costs, and the $\frac{\whi_\Gm-\wlo_\Gm}{\whi_\Gm+\wlo_\Gm}$ of Theorem~\ref{th:inner} becomes active. The fluctuations are a discretisation effect: with few items, adding a single item changes which group is the boundary group $\Gm$, how many members it has, and where the leftover share $c^\star$ falls inside $[0, F^w_\Gm)$. In particular, whenever $C$ lands close to a saturation threshold $F^w_{\le k}$ the leftover $c^\star$ is close to zero, and the loss collapses, as noted at the end of the proof of Theorem~\ref{th:localise-master}, whereas a leftover in the interior of the segment incurs a loss close to the bound of Theorem~\ref{th:tight}.

\noindent
For $n\ge10^3$, the curve flattens and stays at about $6\%$ up to $n=10^6$. This plateau is exactly the behaviour predicted by~\eqref{eq:relative-bound}. Once $n$ is large, the number of groups $m$ is fixed by the covering number of the attribute space at scale $\delta$, so the boundary group grows linearly with $n$: $U_\Gm=\Theta(n/m)$. At the same time, the budget is a fixed fraction of $F^w_{[n]}$ and therefore also grows linearly in $n$, and the clearing price $\nu^\star$ converges to a fixed quantile of the ratio distribution. The ratio $U_\Gm/C$ in~\eqref{eq:relative-bound} thus tends to a constant, and the relative loss stabilises at
$
    \frac{E(C)}{V^\star(C)}\ \lesssim\ \frac{U_\Gm}{C}\cdot\frac{\vrep_\Gm}{\nu^\star}\Bigl(\frac{\whi_\Gm-\wlo_\Gm}{\whi_\Gm+\wlo_\Gm}+\frac{\varepsilon_v}{\vrep_\Gm}\Bigr)
$,
a term that depends on $\delta$ but not on $n$. 
Note that the plateau does not contradict the vanishing per-item bound~\eqref{eq:asymptotic-vanish}, which is stated for a group-size cap $K$; here $K=\Theta(n)$ because $\delta$ is held fixed while $n$ grows, so $K/n$ is itself a constant. If instead $\delta$ were shrunk with $n$ so as to keep $K$ bounded, the per-item loss would decay as $\mathcal{O}(K/n)$. The level of the plateau is itself the effect of $\delta$. It is the loss one accepts for refusing to rank items whose attributes differ by less than $\delta$, and it is the term that Theorem~\ref{th:stability} trades against the unbounded sensitivity of Proposition~\ref{prop:eps}. 
% The next experiment shows how this effect is controlled.

\subsection{Sensitivity to the grouping tolerance}\label{sec:eval-delta}
\noindent
Figure~\ref{fig:error_rate_tolerance} fixes $n=100$ and sweeps $\delta$ over more than three orders of magnitude, with both axes on a logarithmic scale. Again, we can distinguish three regimes. For the smallest $\delta$s, the loss is on the order of $10^{-3}\%$. Most groups are singletons or pairs of nearly identical items, so $\whi_\Gm-\wlo_\Gm\le2L_w\delta$ is tiny, and the bound of Theorem~\ref{th:inner} is correspondingly small. In this regime the grouped allocation is, for all practical purposes, the exact optimum
% , while still enjoying the Lipschitz stability of Theorem~\ref{th:stability} inside the margin $\gamma$.

\noindent
Over the central part of the range, the curve is close to a straight line in $\log-\log$ coordinates, i.e., the loss grows polynomially with $\delta$. The slope is noticeably steeper than one. This is consistent with the structure of the bound~\eqref{eq:localised-error}, in which two factors depend on $\delta$ simultaneously: the term $\frac{\whi_\Gm-\wlo_\Gm}{\whi_\Gm+\wlo_\Gm}\le L_w\frac{\delta}{\wmin}$ grows linearly with $\delta$, and the $U_\Gm$ of the boundary group also grows.
% because coarser cells contain more items. The product of the two terms is therefore superlinear in $\delta$ once the cells are populated by more than one item. The $\mathcal{O}(\delta)$ statement of Theorem~\ref{th:inner} refers to a group of fixed composition; the empirical exponent reflects the additional growth of $|\Gm|$ with $\delta$. In absolute terms, the loss stays below $1\%$ for tolerances up to about $\delta\approx0.1$ on these instances, which provides a practical operating range: $\delta$ can be set to the measurement resolution of the data without a material sacrifice in objective value.

\noindent
Beyond $\delta\approx1$, the curve bends and saturates at roughly $30\%$. At this scale, the grouping collapses to $m=1$, the boundary group is $\Gm=[n]$, and the grouped allocation degenerates into a plain water-filling allocation of the whole instance. The loss then equals $\Delta_{[n]}(C)$, which no longer depends on $\delta$, and it is bounded by Theorem~\ref{th:inner} with the global extremes $\wlo=\wmin$ and $\whi=w_{\max}$. In this case, we forgo all ranking information. 
% Between the two extremes, the curve interpolates smoothly, confirming that $\delta$ acts as a single continuous knob between the exact but discontinuous Dantzig solution ($\delta\to0$) and the fully aggregated water-filling solution ($\delta\to\infty$).

\begin{figure}[!t]
    \centering
    \includegraphics[width=0.6\linewidth]{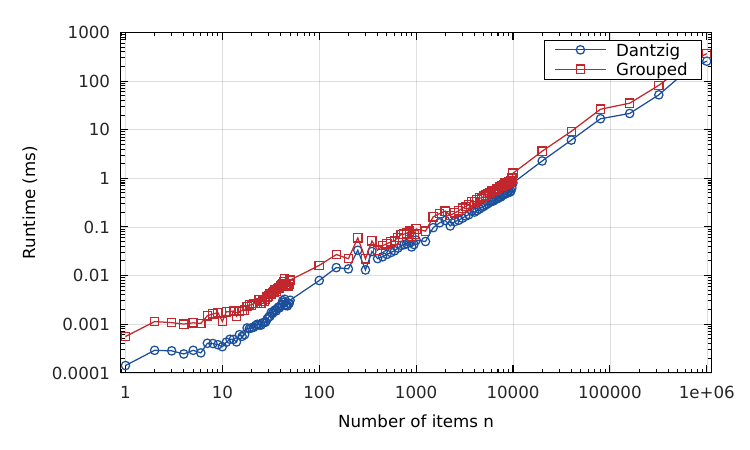}
    \caption{Wall-clock runtime of Dantzig's greedy rule and of the grouped allocation (Algorithms~\ref{alg:grouped-allocation} and~\ref{alg:water-level}) as a function of $n$ (log--log axes). The grouped method carries a constant overhead on small instances and matches the greedy rule on large ones.}
    \label{fig:runtime_comparison}
\end{figure}

\subsection{Runtime}\label{sec:eval-runtime}
\noindent
Figure~\ref{fig:runtime_comparison} compares the runtime of the two methods as $n$ ranges from $1$ to $10^6$, using a moderate $\delta$ for the grouping.
Both curves are straight lines with nearly identical slopes on the $\log-\log$ axes. The grouped method lies above Dantzig's rule by some factor on small instances, and the gap closes as $n$ increases; for $n\ge10^4$ the two curves are almost indistinguishable, and for the largest instances they differ by less than a factor of two.

The behaviour on small instances is explained by the fixed overhead of Stage~1 of Algorithm~\ref{alg:grouped-allocation}: computing the $m$ representative ratios and demands $F^w_{G_k}$, sorting the groups, accumulating the prefix sums $F^w_{\le k}$, and locating $k^\star$. For $n\le10$, the overhead is visible as a vertical offset. The convergence on large instances follows from Theorem~\ref{th:complexity}. The only super-linear terms in the running time of the grouped method are $\mathcal{O}(m\log m)$ for sorting the groups and $\mathcal{O}(|\Gm|\log|\Gm|)$ for sorting the ceilings of the boundary group, and both $m$ and $|\Gm|$ are much smaller than $n$; every other step is a single linear pass over the item arrays. Dantzig's rule sorts all $n$ ratios, but for $n\le10^6$ the $\log n$ factor is at most $20$ and is absorbed into the constant factors of the memory traffic that both implementations share. The asymptotic advantage of $\mathcal{O}(n+m\log m)$ over $\mathcal{O}(n\log n)$ therefore does not yet translate into a measurable speed-up at this scale, but neither does the two-stage structure impose any asymptotic penalty. 
% Replacing the sort of the ceilings in Algorithm~\ref{alg:water-level} by a weighted-median selection, as discussed in Theorem~\ref{th:complexity}, would remove the last super-linear term but would not change the picture, since $|\Gm|\log|\Gm|$ is already negligible.

\section{Summary}\label{sec:eval-summary}
\noindent
The experiments confirm the three claims. The loss is localised and controlled by $\delta$: it is exactly zero when the boundary group is a singleton, and for a fixed $n$ it decreases polynomially as $\delta$ is refined, reaching negligible levels at fine tolerances (Figure~\ref{fig:error_rate_tolerance}). The loss is a function of $\delta$, not of $n$: for a fixed $\delta$, the relative loss settles at a constant plateau as $n\to\infty$ (Figure~\ref{fig:error_rate_vs_n}), matching the $n$-independent form of~\eqref{eq:relative-bound}. The plateau height is the price paid for the Lipschitz stability of Theorem~\ref{th:stability}, which Dantzig's rule cannot offer (Proposition~\ref{prop:eps}). The grouped allocation has the same empirical growth rate as Dantzig's rule and matches its runtime on large instances (Figure~\ref{fig:runtime_comparison}), in line with Theorem~\ref{th:complexity}.
In practice, the tolerance should be chosen at the resolution of the data, i.e., at the scale below which attribute differences are indistinguishable from noise. At such a tolerance, the experiments show a loss well below one per cent, while the allocation gains a bounded perturbation quotient of $\frac{K}{\wmin}$ in place of the unbounded quotient of the exact solution. The source code of the paper is available at \url{https://github.com/cabhinaba3/Grouped-Knapsack}